\documentclass[12pt, reqno]{amsart}
\usepackage{amsfonts, amsmath, amssymb, youngtab, color}
\usepackage[margin=1in]{geometry}
\usepackage{longtable, fancyvrb}
\usepackage[bookmarks, colorlinks=true, linkcolor=blue, citecolor=blue, urlcolor=blue]{hyperref}

\input xy
\xyoption{all}

\allowdisplaybreaks

\numberwithin{equation}{section}

\newtheorem{theorem}{Theorem}[section]
\newtheorem{lemma}[theorem]{Lemma}

\newtheorem{corollary}[theorem]{Corollary}
\newtheorem{proposition}[theorem]{Proposition}
\newtheorem{conjecture}[theorem]{Conjecture}

\theoremstyle{definition}
\newtheorem{example}[theorem]{Example}
\newtheorem{remark}[theorem]{Remark}

\newtheorem{definition}[theorem]{Definition}

\newtheorem{notation}[theorem]{Notation}

\def\CC{{\mathbb C}}
\def\ZZ{{\mathbb Z}}

\DeclareMathOperator{\Tor}{Tor}
\DeclareMathOperator{\reg}{reg}
\DeclareMathOperator{\Sym}{Sym}

\newcommand{\la}{\langle}
\newcommand{\ra}{\rangle}

\newcommand{\OO}{\mathcal{O}}

\newcommand{\bF}{\mathbf{F}}

\newcommand{\arxiv}[1]{\href{http://arxiv.org/abs/#1}{{\tt arXiv:#1}}}

\newcommand{\bd}{\begin{picture}(8,6)
\put(4,-1){\line(0,1){8}}
\put(4,3){\circle*{6}}
\end{picture}}
\newcommand{\nb}{\begin{picture}(8,6)
\put(4,-1){\line(0,1){8}}
\put(3,3){\line(1,0){2}}
\end{picture}}
\newcommand{\vd}{\begin{picture}(8,10)
\put(4,5){\circle*{1}}
\put(4,2){\circle*{1}}
\put(4,8){\circle*{1}}
\end{picture}}

\newlength\tmpl

\begin{document}

\title[Jack polynomials as fractional quantum Hall states]{Jack polynomials as fractional quantum Hall states and the Betti numbers of the $(k+1)$-equals ideal}

\date{March 17, 2013}

\author{Christine Berkesch Zamaere}
\address{Department of Mathematics, Duke University, Durham, NC, USA}
\curraddr{School of Mathematics, University of Minnesota, Minneapolis, MN, USA}
\email{cberkesc@math.umn.edu}

\author{Stephen Griffeth}
\address{Instituto de Matem\'atica y F\'isica, Universidad de Talca, Talca, Chile} \email{sgriffeth@inst-mat.utalca.cl}

\author{Steven V Sam}
\address{Department of Mathematics, University of California, Berkeley, CA, USA}
\email{svs@math.berkeley.edu}

\thanks{SG acknowledges the financial support of Fondecyt Proyecto Regular 1110072. SS was supported by a Miller research fellowship.}

\begin{abstract} 
We show that for Jack parameter $\alpha=-(k+1)/(r-1)$, certain Jack polynomials studied by Feigin--Jimbo--Miwa--Mukhin vanish to order $r$ when $k+1$ of the coordinates coincide. This result was conjectured by Bernevig and Haldane, who proposed that these Jack polynomials are model wavefunctions for fractional quantum Hall states. Special cases of these Jack polynomials include the wavefunctions of Laughlin and Read--Rezayi. In fact, along these lines we prove several vanishing theorems known as \emph{clustering} properties for Jack polynomials in the mathematical physics literature, special cases of which had previously been conjectured by Bernevig and Haldane. Motivated by the method of proof, which in case $r=2$ identifies the span of the relevant Jack polynomials with the $S_n$-invariant part of a unitary representation of the rational Cherednik algebra, we conjecture that unitary representations of the type A Cherednik algebra have graded minimal free resolutions of Bernstein--Gelfand--Gelfand type; we prove this for the ideal of the $(k+1)$-equals arrangement in the case when the number of coordinates $n$ is at most $2k+1$. In general, our conjecture predicts the graded $S_n$-equivariant Betti numbers of the ideal of the $(k+1)$-equals arrangement with no restriction on the number of ambient dimensions. 
\end{abstract}

\maketitle
\section{Introduction}
\label{sec:intro}

The purpose of this paper is to bring tools from the representation theory of Cherednik algebras to bear on problems arising from the physics of the quantum fractional Hall effect and combinatorial commutative algebra. More specifically, we prove conjectures of Bernevig and Haldane~\cite[page 4, first full paragraph]{BeHa08} and~\cite[Section III A,B]{BeHa2} on the order of vanishing of certain Jack polynomials, with Jack parameter equal to a negative rational number $\alpha=-(k+1)/(r-1)$, along certain highly symmetric linear subspace arrangements and study minimal free resolutions of the ideals of these arrangements. The arrangements which appear are generalizations of the \emph{$(k+1)$-equals arrangement}, the set of points in $\CC^n$ where some $k+1$ coordinates are equal. 

In more detail, let $\alpha$ be the Jack parameter. See Section~\ref{sec:symm} for the definition of the Jack polynomials $p^{(\alpha)}_\lambda \in \CC[x_1,x_2,\dots,x_n]$, indexed by partitions $\lambda$ with at most $n$ parts; these are symmetric polynomials whose coefficients are rational functions in $\alpha$. We now state our first main result, which follows from Theorem~\ref{thm:mainthm1}, and in particular proves conjectures of Bernevig and Haldane~\cite{BeHa08,BeHa2}. In fact, it should be regarded as interpolating between the conjectures made in \cite[Section III A,B]{BeHa2}, and so predicts a host of new vanishing properties, in addition to enlarging the set of Jack polynomials appearing in the original conjectures. It describes the behavior of certain classes of Jack polynomials upon allowing some of the coordinates to cluster together. 

\begin{theorem} 
\label{thm:intro}
Suppose $\alpha=-\frac{k+1}{r-1}$ for positive coprime integers $r-1$ and $k+1$ with $1 \leq k < n$, and let $s$ be a positive integer with $s(k+1) \leq n$. Let $\lambda \in \ZZ_{\geq 0}^n$ be a partition such that 
\begin{align}
\lambda_i \geq \lambda_{i+k}+r \quad 
&\hbox{for all $1 \leq i \leq n-(s(k+1)+k)+1$,}
\nonumber
\\
\lambda_i \geq \lambda_{i+k}+s(r-1)+1 \quad 
&\hbox{for all $n-(s(k+1)+k)+1 < i \leq n-(s(k+1)-1)$,}
\text{ and}
\nonumber
\\
\lambda_i \leq \lambda_{i+k+1}+r-1 \quad 
&\hbox{for all $n-(s(k+1)-1)<i \leq n-(k+1)$.} \label{no poles}
\end{align} 
Then the Jack polynomial $p^{(\alpha)}_\lambda(x_1,x_2,\dots,x_n)$ is well-defined, and for each divisor $d$ of $s$,
specializing the first $d(k+1)$ variables to $z_1$, the next $d(k+1)$ variables to $z_2$, and so on, up to the $(s/d)$th cluster of $d(k+1)$ variables, for which we set the first $d(k+1)-1$ equal to $z=z_{s/d}$, the specialized function $p^{(\alpha)}_\lambda(z_1,\dots,z_1,z_2,\dots,z_2,\dots,z,\dots,z,x_{s(k+1)},\dots,x_n)$ is divisible by 
\[
\prod_{i=s(k+1)}^n (x_i-z)^{d(r-1)+1}.
\] 
\end{theorem} The theorem asserts that the Jack polynomials satisfying its hypotheses vanish to order $r$ upon forming $s$ clusters of $k+1$ variables each, and that if one forms instead larger clusters of $d(k+1)$ variables, the order of vanishing increases by roughly a factor of $d$, to $d (r-1) +1$. For example, when $n=10$, $k+1=5$, and $r-1=3$, the $s=1$ case of Theorem~\ref{thm:intro} asserts that the Jack polynomial indexed by 
$\lambda = (8,8,4,4,4,4,0,0,0,0)$ vanishes to order $4$ after the appropriate specialization. The original conjectures of Bernevig and Haldane obtain by replacing condition~\eqref{no poles} by the condition $\lambda_i=0$ for $i>n-(s(k+1)-1)$, and taking $d=1$ or $d=s$. Some special cases of Theorem~\ref{thm:intro} were established by Baratta and Forrester, who obtained a factorization formula for certain Jack polynomials and conjectured an analog for certain Macdonald polynomials~\cite{BaFo}. Dunkl and Luque subsequently proved the Baratta-Forrester conjecture in ~\cite{DuLu}.

In case $s=1$, the class of partitions specified in Theorem~\ref{thm:intro} is precisely the set of \emph{$(k,r,n)$-admissible} partitions studied in \cite{FJMM}. In case $s>1$, the set of partitions allowed by the hypotheses strictly includes the set of $(k,r,s,n)$-admissible partitions of \cite{BeHa2}, which should correspond to Jack polynomials relevant for the construction of fractional quantum Hall quasiparticle excitations. We note that the $r-1=s=1$ case corresponds to a unitary Cherednik algebra module $L_{1/(k+1)}(\lambda)$ for a certain partition $\lambda$, so the space of states carries a positive definite Hermitian form compatible with the Cherednik algebra action and in particular making the usual Calogero--Moser Hamiltonian self-adjoint. (This Hermitian pairing is not the usual one for the Calogero--Moser system, but rather comes from the identification of the state space with $L_{1/(k+1)}(\lambda)$.) 

Our proof proceeds by first establishing an analogous result for non-symmetric Jack polynomials (which are presumably spin fractional quantum Hall wavefunctions) in Theorem~\ref{thm:non-symm vanishing}. To do this, we combine results of Dunkl~\cite{Dun} and Etingof--Stoica~\cite{EtSt} to produce a few non-symmetric Jack polynomials that satisfy an analogue of Theorem~\ref{thm:intro}. We continue with an analysis of when poles may appear in the Knop--Sahi recursion \cite{KnSa} to construct a broader class of non-symmetric Jack polynomials satisfying the same vanishing conditions. Symmetrizing these produces the Jack polynomials appearing in Theorem~\ref{thm:intro}. Finally, we use (an elementary version of) the method of Bezrukavnikov--Etingof induction and restriction functors in~\cite{BE-induction} to obtain a version of the result that depends on a divisor $d$ of $s$. We now state a consequence of this method of proof (see Lemma~\ref{lem:non symm version of IIIB}). 

\begin{corollary}
Let $X_{s,k+1}$ be the set of points in $\CC^n$ having $s$ clusters of $k+1$ equal coordinates. Then every polynomial function $f$ vanishing on $X_{s,k+1}$ vanishes to order at least $s$ at each point of the $s(k+1)$-equals arrangement $X_{1,s(k+1)} \subseteq X_{s,k+1}$.
\end{corollary} 

We are not aware of a proof of this result that does not involve representation theory and non-symmetric Jack polynomials.

It seems likely that the Macdonald polynomial analogue of our result is amenable to the same techniques, although there are some technical obstacles to overcome; the papers~\cite{Kas} and~\cite{Eno} provide a starting point. We do not have a method to prove the superspace analogue of the clustering phenomenon observed in~\cite{DLM}; perhaps a super Cherednik algebra waits to be discovered.

In the last two sections of the paper, Sections~\ref{sec:BGG conj} and~\ref{sec:proof}, we state a conjectural formula for a BGG-style resolution of unitary modules of the type A Cherednik algebra (in Conjecture~\ref{conj:BGG conjecture}) and prove it in some special cases. From now on, in order to conform to standard notation in the Cherednik algebra literature, we set $c:=-1/\alpha$ and will put $m:=k+1$ and $\ell:=r-1$. The main result of these sections is Theorem~\ref{thm:pureres}, a consequence of which is the construction of a minimal pure free resolution for the ideal of the $m$-equals arrangement in case $2m \ge n+1$. (See Section~\ref{sec:proof} for a definition of a pure resolution.) 

\begin{theorem} 
\label{thm:intro2}
If $2m > n+1$, then the coordinate ring of the $m$-equals arrangement has a pure resolution with degree sequence $(0,n-m+1,n-m+2, \dots,m-1,m+1,m+2,\dots,n-1,n)$. Thus the defining ideal is resolved by two linear strands separated by a quadratic jump. If $2m=n+1$, then the defining ideal has a linear resolution. In both cases, this resolution is a complex of modules for the rational Cherednik algebra $H_{1/m}(S_n,\CC^n)$.
\end{theorem}

Conjecture~\ref{conj:BGG conjecture} asserts that every unitary module for the type A Cherednik algebra has a resolution by modules that are sums of standard modules (and predicts exactly which modules occur, all with multiplicity one). As a consequence, we predict a precise combinatorial formula for the graded, $S_n$-equivariant Betti numbers of the ideal of the $m$-equals arrangement, with no restrictions on $m$. In particular, the truth of the conjecture would imply that the projective dimension of the coordinate ring of the $m$-equals arrangement is 
\[
(m-2)e+1,
\] 
where the integer $e$ is the total number of empty spaces lying on a runner above a bead in the $m$-abacus diagram for the partition $\lambda=((m-1)^q,r)$ of $n$, determined by dividing $n$ by $m-1$ to obtain a quotient $q$ and remainder $r$ (see Definition~\ref{def:abacus}). We note that in this formula, $e>1$ if $m \leq n/2$, so we predict that the $m$-equals arrangement is Cohen--Macaulay exactly in the cases covered by Theorem~\ref{thm:intro2}, that is, for $m > n/2$, or if $m=2$. This prediction is implied by \cite[Proposition 3.11]{EGL}.

We end this introduction by mentioning an analogous situation that exists for determinantal ideals. Consider the space $X$ of $n \times m$ matrices as an affine space and let $X_r$ be the locus of matrices of rank $\le r$. It is clear that there is an action of the Lie algebra $\mathfrak{gl}_n \times \mathfrak{gl}_m$ on its coordinate ring, and in fact, this action extends to an action of the larger Lie algebra $\mathfrak{gl}_{m+n}$ in such a way that the coordinate ring is an irreducible unitarizable highest weight module. In \cite{enrighthunziker}, it is observed that the BGG resolutions constructed in \cite{EnWi} are minimal free resolutions of the coordinate ring of $X_r$ over the coordinate ring of $X$. This provides a representation-theoretic interpretation of the resolution discovered by Lascoux in \cite{Las}. A similar story emerges if we replace $X$ with the space of symmetric or skew-symmetric matrices, or, for irreducible Hermitian symmetric pairs in general, with certain versions of determinantal varieties associated to Wallach representations \cite{EH-exceptional}.

\subsection*{Outline.}
In Section~\ref{sec:definitions}, we provide preliminary tools and results. 
Section~\ref{sec:symm} contains our main results on Jack polynomials, Theorem~\ref{thm:non-symm vanishing} and Theorem~\ref{thm:mainthm1}, which yield Theorem~\ref{thm:intro}.  
We state our conjecture regarding BGG-style resolutions of unitary modules of the type A Cherednik algebra in Section~\ref{sec:BGG conj}, and prove a special case, Theorem~\ref{thm:intro2}, in Section~\ref{sec:proof}.

\subsection*{Acknowledgments.} 
We thank Patrick Desrosiers, Jessica Gatica, and Luc Lapointe for teaching us about clustering properties of Jack polynomials, Bernard Leclerc for explanations about the character formula for affine parabolic category $\OO$, Catharina Stroppel for helpful pointers on BGG resolutions, and Charles Dunkl and Peter Forrester for useful comments on a preliminary version of this article. We thank Alexander Kleshchev for bringing Oliver Ruff's work to our attention and Ivan Cherednik for explaining how to construct BGG-style resolutions for certain affine Hecke algebra modules.

\section{The polynomial representation of the rational Cherednik algebra} 
\label{sec:definitions}
 
In this section we first define the Cherednik algebra of a finite reflection group $W$. See \cite[\S 3]{etingofma} for an exposition. We are primarily interested in the case $W=S_n$, as our goal for the section is to derive some first consequences of results of Etingof and Stoica from~\cite{EtSt}.

Let $V$ be a finite dimensional complex vector space, and let $W \subset \mathrm{GL}(V)$ be a finite subgroup such that it is generated by the set
\[
T:=\{s \in W \ | \ \mathrm{codim}(\mathrm{fix}(s))=1 \},
\]
where $\mathrm{fix}(s) := \{v \in V \mid s(v) = v\}$. 
For each $s \in T$ choose $\alpha_s \in V^*$ such that the zero set of $\alpha_s$ is the fixed space of $s$, and let $c_s \in \CC$ be complex numbers with $c_{w s w^{-1}}=c_s$ for all $s \in T$ and $w \in W$. Let $\CC[V]$ be the ring of polynomial functions on $V$, and for each $y \in V$ define a \emph{Dunkl operator} by the formula
\[
y(f):=\partial_y(f)-\sum_{s \in T} c_s \la \alpha_s,y \ra \frac{f-s(f)}{\alpha_s} \quad \hbox{for $f \in \CC[V]$.}
\] 
This formula depends on the choice of $c_s$ but not on the choice of $\alpha_s$.
 
\begin{definition}
The \emph{rational Cherednik algebra} corresponding to these data is the subalgebra $H_c(W,V) \subseteq \mathrm{End}_\CC(\CC[V])$ generated by $W$, $\CC[V]$, and the Dunkl operators for all $y \in V$. We call $\CC[V]$ the \emph{polynomial representation} of $H_c(W,V)$. 
\end{definition}

\begin{notation}
When $W=S_n$ and $V=\CC^n$ is the permutation representation of $S_n$, we use respective coordinates $x_1,x_2,\dots,x_n$ and $y_1,y_2,\dots,y_n$. Here, the $x_i$'s are a basis of $V^*$ and the $y_i$'s are the dual basis of $V$, thought of as acting by Dunkl operators on $\CC[V]=\CC[x_1,x_2,\dots,x_n]$. 
\end{notation}
 
Consider now the case that $W:=S_n$ is acting on 
\[
V_n:=\{(p_1,p_2,\dots,p_n) \in \CC^n \ | \ p_1+p_2+\cdots+p_n=0 \}
\] by permuting coordinates. There is only one conjugacy class of reflections, given by the transpositions, so the parameter $c$ is a complex number. In the special case that $c=m/n$ for a positive integer $m$ coprime to $n$,~\cite[Theorems 1.2 and 1.6]{BEG} imply that there is an ideal $I_{m,n} \subseteq \CC[V_n]$ that is $H_{m/n}(S_n,V_n)$-stable, whose quotient $\CC[V_n]/I_{m,n}$ has dimension $m^{n-1}$, and is generated in degree $m$ by an $S_n$-equivariant regular sequence spanning a copy of the reflection representation of $V_n$. In particular, it follows that $I_{m,n} \subseteq \mathfrak{m}^m$, where $\mathfrak{m}$ is the ideal of functions vanishing at $0$. Thus, elements of $I_{m,n}$ vanish to order $m$ at $0$. 

\begin{definition}
\label{def:m-annihilate}
Given a subset $X \subseteq \CC^n$, let $I$ be the ideal of functions vanishing on $X$. 
If $f \in \CC[x_1,x_2,\dots,x_n]$ has $f \in I^\ell$, then we say that $f$ \emph{$\ell$-annihilates $X$}. 
\end{definition}

Note that the $\ell$-annihilation condition is stronger than requiring only that $f$ vanish to order $\ell$ at each point of $X$. This distinction will be important for us in the sequel.
 
In the remainder of this section, we describe the results of Etingof and Stoica regarding the structure of the polynomial representation of $H_{\ell/m}(S_n,\CC^n)$ in terms of the ideals $I_{\ell,m}$. 

\begin{notation}
\label{not:submodules def} 
Write $n = qm+r$ with nonnegative integers $q,r$ such that $0 \leq r < m$. Observe that $\CC[V_m]=\CC[x_1-x_2,x_2-x_3,\dots,x_{m-1}-x_m]$ embeds in $\CC[x_1,\dots,x_n]$ in $q$ obvious ways, corresponding to the first $m$ coordinates, the next $m$, and so on, so that the $k$th embedding is given by 
\[
x_i-x_{i+1} \mapsto x_{(k-1)m+i}-x_{(k-1)m+i+1}.
\]
For each $1 \leq s \leq q$, let $J_{s,\ell,m}$ be the ideal in $\CC[x_1,\dots,x_n]$ generated by the images of $I_{\ell,m}$ via the first $s$ of these embeddings. Since the set of common zeros of $I_{\ell,m}$ in $V_m$ consists exactly of the origin, on which each element of $I_{\ell,m}$ vanishes to order at least $\ell$, the set of common zeros of $J_{s,\ell,m}$ is
\begin{align*}
Z_{s,m}:=
\left\{
\begin{array}{l}
x_1=x_2=\cdots=x_m,\\
x_{m+1}=x_{m+2}=\cdots=x_{2m},\\
\ \ \vdots\\
x_{(s-1)m+1}=x_{(s-1)m+2}=\cdots=x_{sm}
\end{array}
\right\},
\end{align*}
and elements of $J_{s,\ell,m}$ $\ell$-annihilate this set, i.e., 
\[
J_{s,\ell,m} \subseteq I(Z_{s,m})^{\ell}.
\]
Now set
\[
I_{s,\ell,m}:=\bigcap_{w \in S_n} w(J_{s,\ell,m}).
\]
The zero set of $I_{s,\ell,m}$ is the set $X_{s,m}$ of points in $\CC^n$ that have at least $s$ different clusters of $m$ coordinates all equal to one another, and elements of $I_{s,\ell,m}$ $\ell$-annihilate each irreducible component of $X_{s,m}$. 
\end{notation}

\begin{example}
When $n=5$, the set $X_{2,2}$ consists of points of the form $(x,x,y,y,z)$ and all possible permutations. Note that in the particular case when $\ell=1$, the ideal $I_{s,1,m}$ is the ideal of the variety of points having $s$ clusters of $m$ variables all equal to one another. Thus its variety is a generalization of the $m$-equals arrangement, which occurs when $s=1$. 
\end{example}
 
The following theorem is due to Etingof and Stoica~\cite{EtSt}; it appears as Lemma 5.13 there, at the beginning of the proof of Theorem 5.10 (see also the proof of Remark 5.11). 
For uniformity of notation, we set 
$I_{0,\ell,m}:=0$ and 
$I_{q+1,\ell,m}:=\CC[x_1,x_2,\dots,x_n]$.

\begin{theorem}[Etingof--Stoica] 
\label{thm:submodule classification}
Every $H_{\ell/m}(S_n,\CC^n)$-submodule of $\CC[x_1,x_2,
\dots,x_n]$ is equal to $I_{s,\ell,m}$ for some $0 \leq s \leq q+1$.
\end{theorem}
 
We now show that the symmetric elements $I_{s,\ell,m}^{S_n}$ of $I_{s,\ell,m}$ actually $(\ell+1)$-annihilate the irreducible components of $X_{s,m}$.

\begin{lemma}
\label{lem:symm lemma}
There is a containment $I_{s,\ell,m}^{S_n} \subseteq I(Z_{s,m})^{\ell+1}$. In other words, the elements of $I_{s,\ell,m}^{S_n}$ $(\ell+1)$-annihilate $Z_{s,m}$. 
\end{lemma}
\begin{proof}
We use the fact that the ideal $I_{\ell,m} \subseteq \CC[V_m]$ is generated in degree $\ell$ by a copy of the reflection representation $V_m$ of $S_m$. Given this, the ideal $J_{s,\ell,m} \supseteq I_{s,\ell,m}$ is generated in degree $\ell$ by $s$ different copies $V_m^{(1)},\dots,V_m^{(s)}$ of this reflection representation, spanned by polynomials 
\[
f_1^{(1)},\dots,f_{m-1}^{(1)},\dots,f_1^{(s)},\dots,f_{m-1}^{(s)}.
\]
Now if $f=\sum_{i,p} g_i^{(p)} f_i^{(p)}$ is $S_n$-invariant, it is in particular invariant for the product $S_m \times \cdots \times S_m$ of $s$ copies of $S_m$, and it follows that we may assume that for each $p$, the polynomials $g_i^{(p)}$, for $1 \leq i \leq m-1$, span a copy of $V_m$ in the variables $x_{(p-1)m+1}, \dots, x_{(p-1)m+m}$. Therefore they are in the ideal generated by the differences $x_{(p-1)m+j}-x_{(p-1)m+j+1}$. From Notation~\ref{not:submodules def}, we have $f_i^{(p)} \in I(Z_{s,m})^\ell$, so it now follows that $f \in I(Z_{s,m})^{\ell+1}$, as desired.
\end{proof}
 
In order to prove Theorem~\ref{thm:intro}, the version of our main theorem stated in the introduction, we need the following simple lemma.

\begin{lemma} 
\label{lem:order of vanishing}
Suppose $f$ is a function that $\ell$-annihilates the set $Z_{s,m}$. Specializing the first $m$ variables to $z_1$, the next $m$ variables to $z_2$, and so on, up to the $s$th cluster of $m$ variables, for which we set the first $m-1$ equal to $z=z_s$, the specialized function $f(z_1,\dots,z_1,z_2,\dots,z_2,\dots,z,\dots,z,x_{sm},\dots,x_n)$ is divisible by $(x_{sm}-z)^{\ell}$. Furthermore, if $f$ $\ell$-annihilates each $w(Z_{s,m})$ for $w \in S_n$, then this specialization is divisible by the product
\[
\prod_{i=sm}^n (x_i-z)^{\ell}.
\]
\end{lemma}
\begin{proof}
Let $g$ be a function vanishing on $Z_{s,m}$. Specializing $g$ as we specialized $f$ gives a function divisible by $x_{sm}-z$. Since $f$ is a linear combination of products of $\ell$ such $g$'s, the first statement of the lemma is proved. Since the functions $(x_i-z)^{\ell}$ are pairwise coprime, the second statement follows as well.
\end{proof} 

We close this section with a lemma relating modules over a Cherednik algebra $H_c(W,V)$ to modules over $H_c(W',V')$, where $W' \subseteq W$ is the subgroup of $W$ fixing some point $p \in V$ (such subgroups are known as \emph{parabolic subgroups}) and $V'$ is an orthogonal complement to the fixed point space $V^{W'}$ with respect to a $W$-invariant positive definite Hermitian form.

\begin{lemma} 
\label{lem:IOU}
Let $I'$ be an $H_c(W',V')$-submodule of $\CC[V']$, and let $J$ be the ideal generated by its image in $\CC[V]=\CC[V'] \otimes_\CC \CC[V^{W'}]$. Then 
\[
I := \bigcap_{w \in W} w(J)
\] is an $H_c(W,V)$-submodule of $\CC[V]$.  
\end{lemma}
\begin{proof}
First, observe that $I$ is evidently a $W$-stable ideal in $\CC[V]$. Next, given $f \in I$, we check that $y(f) \in J$ for all $y \in V$. Write $y=y_0+y_1$ with $y_0 \in V'$ and $y_1 \in V^{W'}$, and let $T' := W' \cap T$ be the set of reflections in $W'$. Writing $f=\sum_i f_i g_i$ with $f_i \in \CC[V^{W'}]$ and $g_i \in I'$, it follows from the definition of $y(f)$ that 
\begin{align*}
y(f)&=\partial_{y_0}(f)-\sum_{s \in T'} c_s \la \alpha_s,y_0 \ra \frac{f-s(f)}{\alpha_s}+\partial_{y_1}(f)-\sum_{s \in T \setminus T'} c_s \la \alpha_s,y \ra \frac{f-s(f)}{\alpha_s} \\
&=\sum_i f_i \left(\partial_{y_0}(g_i)-\sum_{s \in T'} c_s \la \alpha_s,y_0 \ra \frac{g_i-s(g_i)}{\alpha_s} \right) + \sum_i g_i \partial_{y_1}(f_i)-\sum_{s \in T \setminus T'} c_s \la \alpha_s,y \ra \frac{f-s(f)}{\alpha_s}.
\end{align*} 
In the last expression, the term in parentheses within the first sum is in $I'$ because $I'$ is $H_c(W',V')$-stable, implying that the first sum is in $J$. The second sum is evidently in $J$, and the last sum is in $J$ because $f \in J$, $s(f) \in J$, and $\alpha_s$ is a nonzerodivisor on $\CC[V]/J$ for all $s \notin T'$. The fact that $\alpha_s$ is a nonzerodivisor on $\CC[V]/J$ may be seen as follows. First, using that $s \notin W'$, we choose a basis $x_1,\dots,x_n$ of $V^*$ so that $\CC[V']=\CC[x_1,x_2,\dots,x_m]$ and $\alpha_s=x_{m+1}$. Now the ideal $J$ is generated by elements of $\CC[x_1,\dots,x_m]$, so $x_{m+1}=\alpha_s$ is a non-zero divisor modulo $J$, establishing the claim. Since we have shown that $y(f) \in J$ for all $y \in V$, it follows that $y(f)=w(w^{-1}(y))(f) \in w(J)$ for all $w \in W$ and $y \in V$.
\end{proof}

\section{Jack polynomials} 
\label{sec:symm}
 
We first describe a partial order on sequences $\mu=(\mu_1,\mu_2,\dots,\mu_n) \in \ZZ_{\geq 0}^n$. Let $\mu^+$ be the non-increasing (partition) rearrangement of $\mu$ and let $w_\mu$ be the permutation with the most inversions such that $w_\mu \mu$ is non-decreasing (an anti-partition). The permutation $w_\mu$ may be described as a ``rank function'' because it ranks the elements of the sequence $(\mu_1,\dots,\mu_n)$ from smallest to largest, where ties are broken by considering as smaller entries farther to the right. 
For example, if $\mu=(4,3,1,1,0,3)$, then $w_\mu=(6,5,3,2,1,4)$.
We say that $\mu < \nu$ if either $\mu^+<\nu^+$ in dominance order, or $\mu^+=\nu^+$ and $w_\mu < w_\nu$ in Bruhat order. We write $\mu^-$ for the non-decreasing rearrangement of $\mu$. 

Let $s_{ij}$ be the transposition that swaps $i$ and $j$ and leaves all other numbers fixed. Let $z_i:=y_i x_i+c \phi_i$, where $\phi_i:=\sum_{1 \leq j < i} s_{ij}$ is the $i$th \emph{Jucys--Murphy--Young element} in $\CC S_n$. For any $c \in \CC$, the subalgebra of $H_c(S_n,\CC^n)$ generated by $z_1,\dots,z_n$ is commutative and acts in an upper-triangular fashion on the polynomial representation $\CC[x_1,\dots,x_n]$, with respect to the basis $x^\mu$ of monomials and the ordering $<$ defined above. This fact is classical, and with our notation, it follows for instance from~\cite[Theorem 5.1]{Gri}. 
Moreover, if we let $c$ be a formal variable, then the polynomial representation $\CC(c)[x_1,\dots,x_n]$ is diagonalizable with respect to the action of $z_1,\dots,z_n$, and we write $f_\mu$ for the simultaneous eigenfunction with leading term $x^\mu$. This is a polynomial in $x_1,\dots,x_n$ whose coefficients are rational functions of $c$. According to~\cite[Theorem 5.1]{Gri}, the eigenvalue is given by
\begin{equation} 
\label{eig equation}
z_i f_\mu=\left(\mu_i+1-(w_\mu(i)-1) c \right) f_\mu.
\end{equation}
 
To explicitly describe lowest weight generators for the submodules constructed by Etingof and Stoica, we follow~\cite{Dun}.  

\begin{notation}
\label{not:partition indexing} 
For each $1 \leq s \leq q$, divide $n-(sm-1)$ by $m-1$ to obtain a quotient $q_s$ and remainder $r_s$, so that 
\[
n-(sm-1)=q_s(m-1)+r_s \quad \hbox{with \,$0 \leq r_s < m-1$.} 
\]
Then set 
\begin{align*}
\tau_{s,m}&:=(sm-1,(m-1)^{q_s},r_s)
\quad \text{and}
\\*
\mu_{s,\ell,m}&:=((\ell(s+q_s))^{r_s},(\ell(s+q_s-1))^{m-1}, (\ell(s+q_s-2))^{m-1},\dots, (\ell s)^{m-1},0^{sm-1}).
\end{align*}
We use partition notation for exponents: an exponent $e$ indicates an entry repeated $e$ times.
\end{notation}
 
\begin{proposition}
\label{prop:Jack generators}
There is an equality 
$I_{s,\ell,m} 
= H_{\ell/m}(S_n,\CC^n) f_{\mu_{s,\ell,m}}$. 
\end{proposition}
\begin{proof}
 According to~\cite{Dun}, when $c=\ell/m$, we have that $y f_{\mu_{s,\ell,m}}=0$ for all $y \in \CC^n$ and the $\CC S_n$-span of the non-symmetric Jack polynomial $f_{\mu_{s,\ell,m}}$ is isomorphic, as a $\CC S_n$-module, to the Specht module $S^{\tau_{s,m}}$. The result now follows from Theorem~\ref{thm:submodule classification} and the observation that the degree of $f_{\mu_{s,\ell,m}}$ decreases as $s$ increases. 
\end{proof}
 
The following result is a weak and non-symmetric version of Theorem~\ref{thm:intro}.

\begin{theorem} 
\label{thm:non-symm vanishing}
Suppose that $c=\ell/m$ for positive coprime integers $\ell$ and $m$ with $m \geq 2$. Suppose that a sequence $\mu$ satisfies the following: 
\begin{itemize}
\item[(a)] 
for $m+1 \leq j \leq sm-1$, 
\[
\mu^-_j \leq \mu^-_{j-m}+\ell, 
\ \ \text{with equality implying that}\ \ 
w_\mu^{-1}(j)> w_\mu^{-1}(j-m),
\]
\item[(b)] 
for $sm \leq j \leq sm+m-2$, 
\[
\qquad\  
\mu^-_j \geq \mu^-_{j-(m-1)}+s\ell, 
\ \ \text{with equality implying that}\ \ 
w_\mu^{-1}(j)< w_\mu^{-1}(j-(m-1)), 
\  \text{and}
\]
\item[(c)] 
for $j \geq sm+m-1$, 
\[
\ \ 
\mu^-_j \geq \mu^-_{j-(m-1)}+\ell, 
\ \ \text{with equality implying that}\ \ 
w_\mu^{-1}(j)< w_\mu^{-1}(j-(m-1)).
\]
\end{itemize}  
Then the non-symmetric Jack polynomial $f_\mu$ 
is well-defined and belongs to the ideal 
$I_{s,\ell,m}$; 
in particular, $f_\mu$ $\ell$-annihilates each 
irreducible component of $X_{s,m}$. 
\end{theorem} 
\begin{proof} 
We describe a recursion that constructs the polynomials $f_\mu$, as in~\cite{KnSa}. Set
\[
\sigma_i := s_i+\frac{c}{z_i-z_{i+1}}, \text{ for $i=1,\dots,n-1$}
\]
and $\Phi := x_n s_{n-1} s_{n-2} \cdots s_1$. For $\mu = (\mu_1,\mu_2,\dots,\mu_n) \in \ZZ^n$, set 
$\phi \mu := (\mu_2,\mu_3,\dots,\mu_n,\mu_1+1)$. Then~\cite[Lemma 5.3]{Gri} shows that
\begin{equation} 
\label{eq:recursion} 
\Phi f_\mu=f_{\phi \mu}, 
\quad \text{and}\quad
\text{if $\mu_i< \mu_{i+1}$, then }  \sigma_i f_\mu=f_{s_i \mu}.
\end{equation}  
For special positive rational values of $c$, the coefficients of $f_\mu$ may have poles. However, assuming that none of the coefficients of $f_\mu$ have a pole at $c \in \CC$, the recursions above show that certain other non-symmetric Jack polynomials will be well-defined at $c$ and belong to the $H_c(S_n,\CC^n)$-submodule generated by $f_\mu$. For instance, $f_{s_i \mu}=\sigma_i f_\mu$ is well-defined if $f_\mu$ is well-defined, provided that the eigenvalues of $z_i$ and $z_{i+1}$ on $f_\mu$ are distinct. In light of~\eqref{eig equation}, it is enough that
\begin{equation}
\mu_i-\mu_{i+1} \neq (w_\mu(i)-w_\mu(i+1))c.
\end{equation}

To describe all of the non-symmetric Jack polynomials that can be obtained in this way starting with $f_\mu$, we make the following definitions. For each $\mu \in \ZZ_{\geq 0}^n$, we write $\mu^-:=w_\mu \mu$ for its non-decreasing rearrangement, and we define a set $R(\mu)$ as follows. 
First, a triple $(i,j,k)$ with $1 \leq i < j \leq n$ and $k \in \ZZ_{>0}$ is in $R(\mu)$ if and only if $\mu^-_j >\mu_i^-+k$ or $\mu^-_j=\mu^-_i+k$ and $w_\mu^{-1}(j) < w_\mu^{-1}(i)$. 
Second, a pair $(i,k)$ with $1 \leq i \leq n$ and $k \in \ZZ_{>0}$ is in $R(\mu)$ if and only if $\mu^-_i \geq k$. 
(These definitions are \cite[(7.17)]{Gri}, specialized to our situation; in fact, we could use the machinery developed there to finish the proof immediately. However, since the combinatorics is somewhat simpler in our situation, the extra generality is unnecessary, and we prefer to outline a self-contained argument.) Finally, a triple $(i,j,k)$ as above is \emph{non-semisimple} if $k=\ell(j-i)/m$.

With these definitions, the hypotheses (b) and (c) on $\mu$ in the statement of the theorem then amount to $R(\mu_{s,\ell,m}) \subseteq R(\mu)$, while (a) implies that $R(\mu) \setminus R(\mu_{s,\ell,m})$ does not contain any non-semisimple triples. We will show that this implies that $f_\mu$ can be obtained from $f_{\mu_{s,\ell,m}}$ by use of the recursions \eqref{eq:recursion}.

Now let $\mu, \nu \in \ZZ_{\geq 0}^n$ be arbitrary. It follows from the definitions that $R(\mu)=R(\nu)$ implies $\mu=\nu$. Moreover, $R((0,0,\dots,0))=\emptyset$ and for all $\mu$,
\[
R(\phi \mu)=R(\mu) \cup \{(w_\mu(1),\mu_1+1)\}
\]
and if $\mu_i < \mu_{i+1}$, then
\begin{equation} 
\label{eq:polelocation} 
R(s_i \mu)=R(\mu) \cup \{(w_\mu(i), w_\mu(i+1),\mu_{i+1}-\mu_i) \}.
\end{equation} 
Note that the recursion \eqref{eq:recursion}, $f_{s_i \mu}=\sigma_i f_\mu$, is well-defined here if and only if 
\[
(w_\mu(i), w_\mu(i+1),\mu_{i+1}-\mu_i)
\] 
is not a non-semisimple triple. Now assume $R(\mu) \subseteq R(\nu)$. One checks that at least one of the following holds: 
\begin{itemize}
\item[(i)] $\mu=\nu$,
\item[(ii)] $(w_\mu(1),\mu_1+1) \in R(\nu)$, or
\item[(iii)] there is some $1 \leq i \leq n-1$ with $\mu_i < \mu_{i+1}$ and $(w_\mu(i),w_\mu(i+1),\mu_{i+1}-\mu_i) \in R(\nu)$.
\end{itemize} Given this, it follows that either $\mu=\nu$, or that $R(\mu) \subsetneq R(\phi \mu) \subseteq R(\nu)$ or that $R(\mu) \subsetneq R(s_i \mu) \subseteq R(\nu)$ for some $1 \leq i \leq n-1$. By induction, one may therefore find a sequence $\mu=\mu^{(0)},\mu^{(1)},\dots,\mu^{(p)}=\nu$ with the following properties: $R(\mu^{(i)})$ is obtained from $R(\mu^{(i-1)})$ by adjoining a single element of $R(\nu)$, and either $\mu^{(i)}=\phi \mu^{(i-1)}$ or $\mu^{(i)}=s_j \mu^{(i-1)}$ for some $1 \leq j \leq n-1$ with $\mu^{(i-1)}_j< \mu^{(i-1)}_{j+1}$. Finally, assuming no element of $R(\nu) \setminus R(\mu)$ is non-semisimple, this last condition implies that the recursion \eqref{eq:recursion} is well-defined at each stage of this process, thanks to \eqref{eq:polelocation}. This completes the proof.
\end{proof} Even in the simplest examples the Jack polynomials in the theorem do not exhaust the set of all well-defined Jack polynomials belonging to $I_{s,\ell,m}$, but for our purposes they are enough.
 
We now show that Theorem~\ref{thm:non-symm vanishing} can be upgraded to a stronger vanishing property for symmetric Jack polynomials. Let $e := \sum_{w \in S_n} w$ be the symmetrizer; this is idempotent up to a factor of $n!$. Given a non-increasing sequence $\lambda=(\lambda_1,\lambda_2,\dots,\lambda_n) \in \ZZ_{\geq 0}^n$, write $\lambda^-$ for its non-decreasing rearrangement and $n_\lambda$ for the order of its stabilizer in the symmetric group $S_n$, which acts by permuting coordinates. Define the \emph{(symmetric) Jack polynomial} $p_\lambda$ to be 
\[
p_\lambda := \frac{1}{n_\lambda} e f_{\lambda^-}.
\] If we wish to emphasize the dependence on $\alpha$ we may write $p_\lambda^{(\alpha)}$, as in the introduction. This has as leading term a monomial symmetric function $m_\lambda$, and it is an eigenfunction for symmetric polynomials in the $z_i$'s; in fact, it is characterized by these properties.
 
Evidently, Theorem~\ref{thm:non-symm vanishing} and Lemma~\ref{lem:symm lemma} together imply a vanishing to the order $\ell+1$ for the symmetrizations of the polynomials appearing in Theorem~\ref{thm:non-symm vanishing}. 
This follows because, given a partition $\lambda$ with at most $n$ nonzero parts such that $\lambda^-$ satisfies the conditions of Theorem~\ref{thm:non-symm vanishing}, the equation $p_\lambda=\frac{1}{n_\lambda} e f_{\lambda^-}$ implies that $p_\lambda \in I_{s,\ell,m}^{S_n}$. The conditions on $\lambda$ in the following theorem guarantee that the anti-partition rearrangement $\lambda^-$ of $\lambda$ satisfies the hypotheses of Theorem \ref{thm:non-symm vanishing}.
\begin{theorem} 
\label{thm:symm vanishing} 
Let $\lambda \in \ZZ_{\geq 0}^n$ be a partition such that 
\begin{align*}
\lambda_i \geq \lambda_{i+m-1}+\ell+1 \quad 
&\hbox{for all $1 \leq i \leq n-(sm+m-1)+1$,}
\\
\lambda_i \geq \lambda_{i+m-1}+s \ell+1 \quad 
&\hbox{for all $n-(sm+m-1)+2 \leq i \leq n-(sm-1)$,}
\text{ and}
\\
\lambda_i \leq \lambda_{i+m}+\ell \quad 
&\hbox{for all $n-(sm-1)<i \leq n-m$.}
\end{align*}
Then the Jack polynomial $p_\lambda$ is well-defined and belongs to $I_{s,\ell,m}$, and so $(\ell+1)$-annihilates each irreducible component of $X_{k,m}$. 
In particular, upon specializing the variables as in Lemma~\ref{lem:order of vanishing}, each such $p_\lambda$ is divisible by $\prod_{i=sm}^n (x_i-z)^{\ell+1}.$
\end{theorem} 

The minimal degree partition satisfying the conditions of Theorem~\ref{thm:symm vanishing} is 
\[
\lambda=((s(\ell+q_s)+q_s+1)^{r_s},(s(\ell+q_s-1)+q_s)^{m-1},\dots,(s(\ell+1)+2)^{m-1},(s \ell+1)^{m-1},0^{sm-1}).
\]
As an example, Theorem~\ref{thm:symm vanishing} asserts that the Jack polynomial $p_{(7,0,0,0)}(x,x,y,z)$ at Jack parameter $-2/3$ is well-defined and divisible by $(y-z)^4$, as are, for example, $p_{(8,1,0,0)}(x,x,y,z)$ and $p_{(14,7,5,5)}(x,x,y,z)$. 

We conclude this section by proving the conjecture of Bernevig and Haldane from~\cite[Section~III~B]{BeHa2}. For this, we let $d$ and $d'$ be positive integers with $d d' \leq q$. 

\begin{lemma} 
\label{lem:non symm version of IIIB}
The elements of the ideal $I_{d d',\ell,m}$ all $(d \ell)$-annihilate each irreducible component of the arrangement $X_{d',dm}$.
\end{lemma}
\begin{proof}
We first treat the special case that the number of variables is $n=dm$ and $d'=1$. By Notation~\ref{not:partition indexing}, the non-symmetric Jack polynomials generating $I_{d,\ell,m}$ have degree $d \ell$, and they are evidently annihilated by the operator $y_1+y_2+\cdots+y_n$, which a direct calculation shows is the sum of the partial derivatives with respect to the $x$ variables. These non-symmetric Jack polynomials are therefore polynomials in the pairwise differences $x_1-x_2,\dots,x_{n-1}-x_n$, establishing the special case.

To treat the general case, we temporarily append the number $n$ (the ambient dimension) to all subscripts, so that for instance $I_{d,\ell,m,n}$ refers to the corresponding ideal in $n$ variables. We have proved that elements of $I_{d,\ell,m,dm}$ all $d \ell$-annihilate $X_{1,dm,dm}$. For $n \geq kd'm$, let $J$ be the ideal generated by the images of $I_{d,\ell,m,dm}$ via the first $d'$ obvious inclusions $\CC[x_1,\dots,x_{km}] \hookrightarrow \CC[x_1,\dots,x_n]$. The elements of the intersection
\[
I:=\bigcap_{w \in S_n} w(J)
\] 
then $d \ell$-annihilate the irreducible components of $X_{d',dm,n}$. On the other hand, the zero set of $I$ is exactly $X_{dd',m,n}$, and by Lemma~\ref{lem:IOU}, $I$ is $H_{\ell/m}(S_n,\CC^n)$-stable. But these properties characterize $I_{dd',\ell,m,n}$ by Theorem~\ref{thm:submodule classification}, implying that $I=I_{dd',\ell,m,n}$.
\end{proof}

\begin{theorem} \label{thm:mainthm1}
The elements of $I_{dd',\ell,m}^{S_n}$ all $(d \ell + 1)$-annihilate each irreducible component of the arrangement $X_{d',dm}$.
\end{theorem}
\begin{proof}
This follows from Lemma~\ref{lem:non symm version of IIIB} via the argument in Lemma~\ref{lem:symm lemma}. In detail, with notation as in Lemma~\ref{lem:symm lemma}, a symmetric element $f \in I_{dd', \ell,m}^{S_n}$ may be written as
\[
f=\sum_{i,p} g_i^{(p)} f_i^{(p)}
\] 
with $g_i^{(p)}$, for $1 \leq i \leq m-1$, spanning a copy of $V_m$ and therefore contained in the ideal generated by the differences $x_{(p-1)m+j}-x_{(p-1)m+j+1}$. This provides the additional vanishing required beyond that supplied by Lemma~\ref{lem:non symm version of IIIB}.
\end{proof} Theorem \ref{thm:intro} now follows from this result together with Theorem \ref{thm:symm vanishing} and Lemma \ref{lem:order of vanishing}, keeping in mind that $k+1=m$ and $r-1=\ell$.

\section{Conjectural BGG-style resolution}
\label{sec:BGG conj}
 
In this section, we present a conjecture about the structure of $I_{1,1,m}$ as a Cherednik algebra module. Cherednik conjectured, and Etingof and Stoica proved, that this ideal is a unitary module~\cite{EtSt}. We conjecture that for the Cherednik algebra of $S_n$, all unitary modules have resolutions by sums of standard modules. This would provide an analogue of the Bernstein--Gelfand--Gelfand resolution for a finite dimensional Lie algebra module, as well as an analogue of a theorem of Cherednik~\cite{Che} for the single affine Hecke algebra. 

It is not difficult to find examples of finite dimensional modules for the Cherednik algebra for groups $W$ other than $S_n$ that do not admit this type of resolution; the smallest example occurs when $W=G(2,1,2)$. On the other hand, one might wonder if unitary modules always have such resolutions; here we present enough evidence to make this conjecture in type A. Using the classification of unitary modules for the rational Cherednik algebras of the groups $G(r,1,n)$ given in~\cite{Gri3}, one might begin to gather enough data to extend this conjecture (or find a counterexample) to other cases.
 
Evidently, the results of Etingof--Stoica and Dunkl together imply that the non-symmetric Jack polynomials in the $\CC S_n$-span of $f_{\mu_{s,1,m}}$ give a minimal set of generators for the ideal $I_{s,1,m}$ of the set of points in $\CC^n$, where there are $s$ clusters of $m$ equal coordinates, generalizing the work of Li and Li~\cite{LiLi} to this case. However, we will see that the case $s=1$ is special in another way: since it corresponds to a unitary Cherednik algebra module, its minimal resolution is (conjecturally) of representation theoretic origins. 
 
The \emph{PBW Theorem} for $H_c = H_c(W,V)$ states that multiplication induces a vector space isomorphism (see~\cite[Corollary 2.2]{Gri2} for an elementary proof), 
\[
\CC[V] \otimes_\CC \CC W \otimes_\CC \CC[V^*] \cong H_c.
\]
The \emph{standard module} $\Delta_c(\lambda)$ corresponding to $S^\lambda$ is defined as
\[
\Delta_c(\lambda):=\mathrm{Ind}_{\CC[V^*] \rtimes W}^{H_c} S^\lambda,
\]  
where $S^\lambda$ is a $\CC[V^*]$-module such that elements of $V$ act by zero.
By the PBW theorem, there is a vector space isomorphism $\Delta_c(\lambda) \cong \CC[V] \otimes_\CC S^\lambda$ identifying $\Delta_c(\lambda)$ with $S^\lambda$-valued polynomial functions on $V$. The action of $y \in V$ is  given by an analogue of the Dunkl operators:
\[
y(f \otimes u)=\partial_y(f) \otimes u-\sum_{s \in T} c_s \la \alpha_s, y \ra \frac{f-s(f)}{\alpha_s} \otimes s(u) \quad \hbox{for $f \in \CC[V]$ and $u \in S^\lambda$.}
\]
 
We fix a positive definite $W$-invariant Hermitian form $(\cdot,\cdot)$ on each irreducible $\CC W$-module $S^\lambda$. In particular, we obtain a conjugate linear $W$-equivariant isomorphism $V^* \rightarrow V$, which induces a $W$-equivariant conjugate linear ring isomorphism $\CC[V] \rightarrow \CC[V^*]$, written $f \mapsto \overline{f}$. We use this to define a Hermitian pairing $(\cdot,\cdot)_c$ on $\Delta_c(\lambda)$, given by the formula
\[
(f_1 \otimes u_1,f_2 \otimes u_2)_c := (u_1,\overline{f_1} (f_2 \otimes u_2)(0) ) \quad \hbox{for $f_1,f_2 \in \CC[V]$ and $u_1,u_2 \in S^\lambda$,}
\] 
where on the right hand side the formula $\overline{f_1} (f_2 \otimes u_2)(0)$ is computed by first allowing the composition of Dunkl operators $\overline{f_1}$ to act on $f_2 \otimes u_2$, and then evaluating the resulting function at $0$ to obtain an element of $S^\lambda$.

Define $L_c(\lambda)$ to be the quotient of $\Delta_c(\lambda)$ by the radical of $(\cdot,\cdot)_c$, 
\[
L_c(\lambda) := 
\Delta_c(\lambda) / \mathrm{Rad}(\cdot,\cdot)_c.
\] 
It is not difficult to see that $L_c(\lambda)$ is the unique simple quotient of $\Delta_c(\lambda)$, or equivalently that $\mathrm{Rad}(\cdot,\cdot)_c=\mathrm{Rad}(\Delta_c(\lambda))$. By definition, the form $(\cdot,\cdot)_c$ descends to a nondegenerate pairing on $L_c(\lambda)$. We say that $L_c(\lambda)$ is \emph{unitary} if this form is positive definite.

From now on, we assume that $W=S_n$ is acting on $V=\CC^n$. 

\begin{definition}
\label{def:abacus}
Given a partition $\lambda$ and a positive integer $m$, we define its $m$-\emph{abacus diagram} as follows: the abacus has $m$ vertical runners labeled $0,1,2,\dots,m-1$, and beads are placed on the runners in a way we will describe. First, tracing the border of the partition $\lambda$ from northeast to southwest produces a sequence of ``down steps" and ``left steps." We now read across the positions in the abacus from left to right and top to bottom, leaving an empty space for each down step, and a bead for each left step. Thus we always begin with an empty space, and the number of beads in the abacus diagram is $\lambda_1$. 
\end{definition}

\begin{example}
\label{ex:abacus}
If $m=5$ and 
\[
\lambda=(4,4,3)=\begin{array}{c} \tiny \young(\hfil\hfil\hfil\hfil,\hfil\hfil\hfil\hfil,\hfil\hfil\hfil) \end{array}, \ \hbox{then the $5$-abacus for $\lambda$ is} \quad
\begin{smallmatrix}
0&1&2&3&4\\[2pt]\hline
\nb&\nb&\bd&\nb&\bd\\
\bd&\bd&\nb&\nb&\nb\\
\nb&\nb&\nb&\nb&\nb\\
\nb&\nb&\nb&\nb&\nb\\
\nb&\nb&\nb&\nb&\nb\\
\vd&\vd&\vd&\vd&\vd
\end{smallmatrix}. 
\] 
\end{example}

Let $m$ be a positive integer. The module $L_{1/m}(\lambda)$ is unitary precisely when, in the abacus diagram for $\lambda$, the first bead and the last bead are separated by at most $m$ positions (and in particular, there is at most one bead on each runner). That is, if the first bead occurs on runner $i>0$ in some row, the last bead must occur on or before runner $i-1$ in the next row (or if $i=0$, on or before the last runner on the same row). This is an immediate consequence of the classification of unitary modules given by the second author and Etingof--Stoica in~\cite{EtSt}. In particular, the unitary module $I_{1,1,m}$ is $L_{1/m}(\lambda)$, where $\lambda=(m-1,m-1,\dots,m-1,r)$ contains a part $m-1$ with multiplicity $q$ and a part $r$ with multiplicity $1$. Here, we have divided $n$ by $m-1$ to obtain a quotient $q$ and remainder $r$. The $\lambda$ in Example~\ref{ex:abacus} corresponds in this way to the $5$-equals set in $\CC^{11}$.

\begin{definition}
\label{def:homological degree}
Let $D$ be an $m$-abacus diagram with at most one bead on each runner. The \emph{homological degree} $\mathrm{hd}(D)$ is defined to be the sum over all pairs of beads $b_1$ on runner $i$ and $b_2$ on runner $j$ in the diagram $D$, with $i<j$, of $k-1$ if the bead $b_1$ is $k>1$ rows lower than the bead $b_2$, and of $k$ if the bead $b_1$ is $k>0$ rows higher than the bead $b_2$. In short, 
\[
\mathrm{hd}(D) := \sum_{\substack{i<j\\ b_1 
\text{ lies below } 
b_2}}(\text{row}(b_1) - \text{row}(b_2)-1) 
\ \ + 
\sum_{\substack{i<j\\ b_1 
\text{ lies above } 
b_2 }} (\text{row}(b_2)-\text{row}(b_1)).
\]
\end{definition}

It is a straightforward check that the homological degree of the $m$-abacus diagram of $\lambda$ is $0$ when $L_{1/m}(\lambda)$ is unitary. This includes the diagram in Example~\ref{ex:abacus}.

\begin{notation}
Let $P_m(\lambda)$ be the set of all partitions obtained as follows: starting with the $m$-abacus diagram for $\lambda$, we are allowed to move any bead up one row provided that we compensate by moving another bead down one row. Composing moves of this type, we obtain all the $m$-abacus diagrams of the elements of $P_m(\lambda)$. 
\end{notation}

We now describe the conjectural BGG-style resolution of a unitary module $L_{1/m}(\lambda)$. 

\begin{conjecture} 
\label{conj:BGG conjecture}
The unitary module $L_{1/m}(\lambda)$ has a resolution 
\[
\cdots \longrightarrow M_i \longrightarrow \cdots \longrightarrow M_0 \longrightarrow L_{1/m}(\lambda) \longrightarrow 0
\] 
where $M_i$ is the sum of all standard modules 
\[
\Delta_{1/m}(\mu)=\CC[x_1,x_2,\dots,x_n] \otimes S^\mu
\] 
over partitions $\mu \in P_m(\lambda)$ whose corresponding $m$-abacus diagram has homological degree $i$. 
\end{conjecture}

\begin{example}
\label{ex:BGG}
When $\lambda=(4,4,3)$ and $m=5$, we list the $m$-abacus diagrams, organized by homological degree from $0$ to $6$, together with the Young diagrams of the corresponding partitions and a statistic $c_\mu$ that controls the degrees of the maps in the resolution, which will be explained in~\eqref{eq:cmu}.

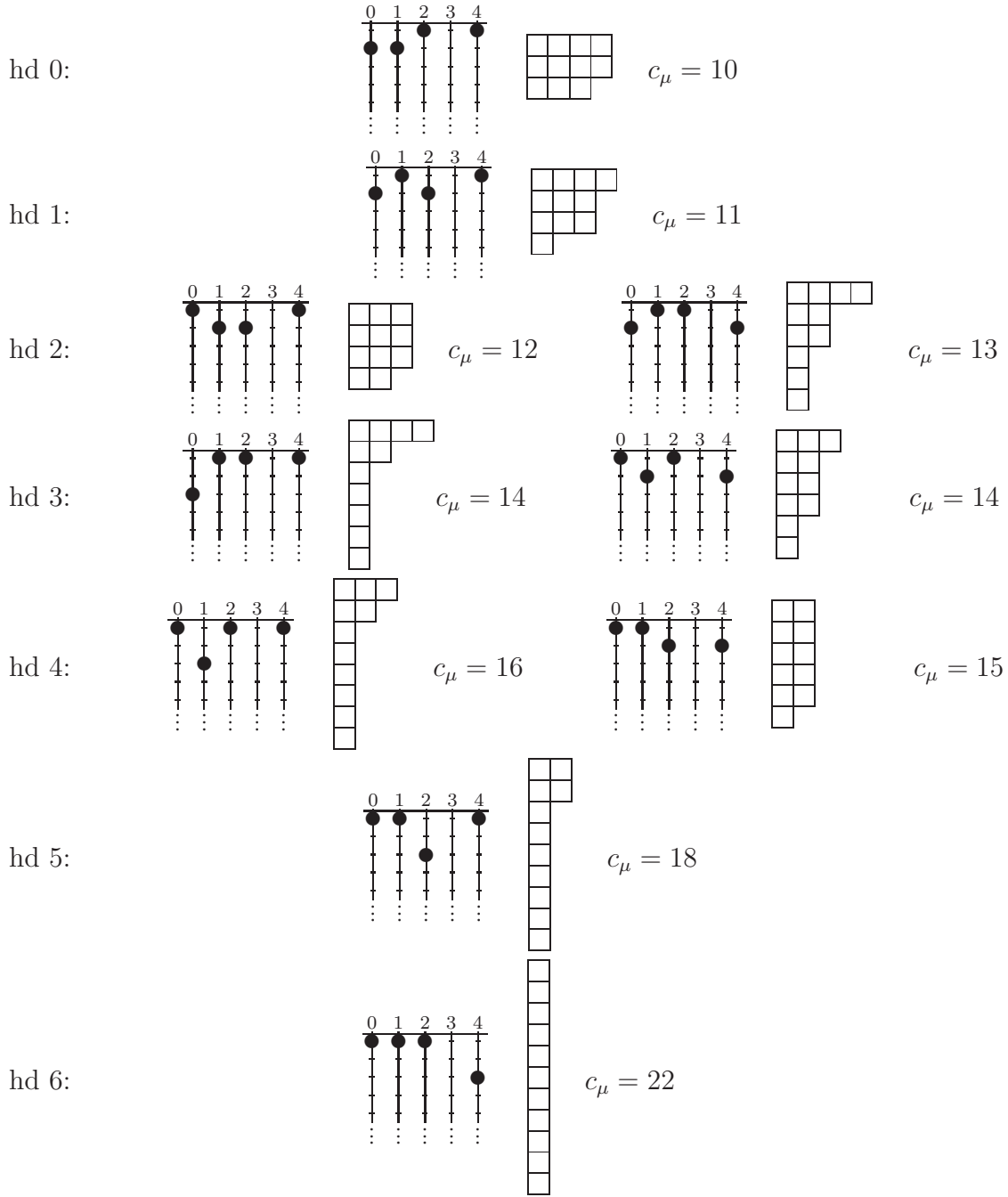
\enlargethispage{\baselineskip}
\begin{longtable}{lc}
hd 0: 
& 
$	\begin{array}{ccc}
	\begin{smallmatrix}
	0&1&2&3&4\\[2pt]\hline
	\nb&\nb&\bd&\nb&\bd\\
	\bd&\bd&\nb&\nb&\nb\\
	\nb&\nb&\nb&\nb&\nb\\
	\nb&\nb&\nb&\nb&\nb\\
	\nb&\nb&\nb&\nb&\nb\\
	\vd&\vd&\vd&\vd&\vd
	\end{smallmatrix} 
	&
\begin{array}{c} \tiny \yng(4,4,3) \end{array}
	&
	c_\mu=10
	\end{array}$
\vspace{2mm}
\\
hd 1:
&
$	\begin{array}{ccc}
	\begin{smallmatrix}
	0&1&2&3&4\\[2pt]\hline
	\nb&\bd&\nb&\nb&\bd\\
	\bd&\nb&\bd&\nb&\nb\\
	\nb&\nb&\nb&\nb&\nb\\
	\nb&\nb&\nb&\nb&\nb\\
	\nb&\nb&\nb&\nb&\nb\\
	\vd&\vd&\vd&\vd&\vd
	\end{smallmatrix}
	&
	\begin{array}{c} \tiny \yng(4,3,3,1) \end{array} 
	&
	c_\mu=11
	\end{array}$
\\
hd 2:
&
\hspace{9.5mm}
$	\begin{array}{cccccc}
	\begin{smallmatrix}
	0&1&2&3&4\\[2pt]\hline
	\bd&\nb&\nb&\nb&\bd\\
	\nb&\bd&\bd&\nb&\nb\\
	\nb&\nb&\nb&\nb&\nb\\
	\nb&\nb&\nb&\nb&\nb\\
	\nb&\nb&\nb&\nb&\nb\\
	\vd&\vd&\vd&\vd&\vd
	\end{smallmatrix}
	&
	\begin{array}{c} \tiny \yng(3,3,3,2) \end{array}  
	& 
	c_\mu=12 \qquad
	&
	\begin{smallmatrix}
	0&1&2&3&4\\[2pt]\hline
	\nb&\bd&\bd&\nb&\nb\\
	\bd&\nb&\nb&\nb&\bd\\
	\nb&\nb&\nb&\nb&\nb\\
	\nb&\nb&\nb&\nb&\nb\\
	\nb&\nb&\nb&\nb&\nb\\
	\vd&\vd&\vd&\vd&\vd
	\end{smallmatrix}
	&
	\begin{array}{c} \tiny \yng(4,2,2,1,1,1) \end{array} 
	& 
	c_\mu=13
	\end{array}$
\\
hd 3: 
&
\hspace{9.5mm}
$	\begin{array}{cccccc}
	\begin{smallmatrix}
	0&1&2&3&4\\[2pt]\hline
	\nb&\bd&\bd&\nb&\bd\\
	\nb&\nb&\nb&\nb&\nb\\
	\bd&\nb&\nb&\nb&\nb\\
	\nb&\nb&\nb&\nb&\nb\\
	\nb&\nb&\nb&\nb&\nb\\
	\vd&\vd&\vd&\vd&\vd
	\end{smallmatrix} 
	&
	\begin{array}{c} \tiny \yng(4,2,1,1,1,1,1) \end{array}
	&
\hspace{-5mm}
	c_\mu=14 \qquad
\hspace{5mm}
	&
\hspace{-4.75mm}
	\begin{smallmatrix}
	0&1&2&3&4\\[2pt]\hline
	\bd&\nb&\bd&\nb&\nb\\
	\nb&\bd&\nb&\nb&\bd\\
	\nb&\nb&\nb&\nb&\nb\\
	\nb&\nb&\nb&\nb&\nb\\
	\nb&\nb&\nb&\nb&\nb\\
	\vd&\vd&\vd&\vd&\vd
	\end{smallmatrix}
	&
	\begin{array}{c} \tiny \yng(3,2,2,2,1,1) \end{array}
	&
\hspace{5mm}
	c_\mu=14
\hspace{-5mm}
\hspace{4.75mm}
	\end{array}$
\\
hd 4: 
	&
$	\begin{array}{cccccc}
	\begin{smallmatrix}
	0&1&2&3&4\\[2pt]\hline
	\bd&\nb&\bd&\nb&\bd\\
	\nb&\nb&\nb&\nb&\nb\\
	\nb&\bd&\nb&\nb&\nb\\
	\nb&\nb&\nb&\nb&\nb\\
	\nb&\nb&\nb&\nb&\nb\\
	\vd&\vd&\vd&\vd&\vd
	\end{smallmatrix}
	&
	\begin{array}{c} \tiny \yng(3,2,1,1,1,1,1,1) \end{array}
	&
	c_\mu=16 \qquad
	&
	\begin{smallmatrix}
	0&1&2&3&4\\[2pt]\hline
	\bd&\bd&\nb&\nb&\nb\\
	\nb&\nb&\bd&\nb&\bd\\
	\nb&\nb&\nb&\nb&\nb\\
	\nb&\nb&\nb&\nb&\nb\\
	\nb&\nb&\nb&\nb&\nb\\
	\vd&\vd&\vd&\vd&\vd
	\end{smallmatrix}
	&
	\begin{array}{c} \tiny \yng(2,2,2,2,2,1) \end{array}
	&
\hspace{9.5mm}
	c_\mu=15
\hspace{-9.5mm}
	\end{array}$
\\
hd 5:
	&
\hspace{-8.5mm}
$	\begin{array}{ccc}
	\begin{smallmatrix}
	0&1&2&3&4\\[2pt]\hline
	\bd&\bd&\nb&\nb&\bd\\
	\nb&\nb&\nb&\nb&\nb\\
	\nb&\nb&\bd&\nb&\nb\\
	\nb&\nb&\nb&\nb&\nb\\
	\nb&\nb&\nb&\nb&\nb\\
	\vd&\vd&\vd&\vd&\vd
	\end{smallmatrix}
	&
	\begin{array}{c} \tiny \yng(2,2,1,1,1,1,1,1,1) \end{array}
	&
	c_\mu=18
	\end{array}$
\\
hd 6: 
	&
\hspace{-12mm}
$	\begin{array}{ccc}
	\begin{smallmatrix}
	0&1&2&3&4\\[2pt]\hline
	\bd&\bd&\bd&\nb&\nb\\
	\nb&\nb&\nb&\nb&\nb\\
	\nb&\nb&\nb&\nb&\bd\\
	\nb&\nb&\nb&\nb&\nb\\
	\nb&\nb&\nb&\nb&\nb\\
	\vd&\vd&\vd&\vd&\vd
	\end{smallmatrix}
	&
	\begin{array}{c} \tiny \yng(1,1,1,1,1,1,1,1,1,1,1) \end{array}
	&
	c_\mu=22
	\end{array}$
\end{longtable}
\end{example}

The set $P_m(\lambda)$ can also be described as the set of all partitions $\mu$ of $n$ with the same $m$-core as $\lambda$ \cite[I.1, Example 8]{macdonald} and $\mu \le \lambda$ in dominance order. Our conjectural BGG-style resolution is a strengthening of a result of Ruff calculating the characters of \emph{completely splittable} irreducible finite Hecke algebra modules~\cite{Ruf}. Using the Knizhnik--Zamolodchikov functor as in \cite{GGOR} shows that Ruff's result follows from our conjecture. 

The degrees of the maps in the conjectural complex above can also be computed using the Cherednik algebra. Maps between objects in category $\OO_c$ are automatically of degree $0$, where the grading is given by the action of the deformed Euler element 
\[
h:=\sum_{1 \le i \le n} x_i y_i + c \sum_{1 \leq i< j \leq n}( 1-s_{ij}).
\] 
The action of $h$ on the degree $d$ part of the standard module $\Delta_c(\mu)=\CC[x_1,\dots,x_n] \otimes S^\mu$ is multiplication by the scalar $d+c(n(\lambda)-n(\lambda^t)+n(n-1)/2)$, where $n(\mu) := \sum_i (i-1)\mu_i$. Set
\begin{align}
\label{eq:cmu}
c_\mu:=c\left(n(\lambda)-n(\lambda^t)+n(n-1)/2\right).
\end{align}
These numbers have been listed above in Example~\ref{ex:BGG} with $c=1/5$. In this example, if we write the $i$th module $M_i$ in the resolution of Conjecture~\ref{conj:BGG conjecture} as a sum of modules of the form 
\[
\left(\CC[x_1,x_2,\dots,x_n] \otimes S^\mu \right)(c_\lambda-c_\mu),
\] 
where the parentheses denote the usual grading shift, then all the maps in the resulting complex are degree $0$. Also, none of the entries of the maps contain nonzero constant terms because the value of the $c_\mu$-function is increasing in the homological degree of the $m$-abacus of $\mu\in P_m(\lambda)$. Thus, Conjecture~\ref{conj:BGG conjecture} implies that our BGG-style resolution is a minimal free resolution of the $m$-equals ideal. Conjecture~\ref{conj:BGG conjecture} also predicts a combinatorial formula for the $S_n$-equivariant graded Betti numbers of the $m$-equals ideal.

In Example~\ref{ex:BGG}, the resolution of the ideal $I$ for the $5$-equals set in $\CC^{11}$ is as follows (we write $\lambda(-d)$ in place of $S^\lambda \otimes_\CC \CC[x_1, \dots, x_n](-d)$):
\begin{align*}
0 &\to (1^{11})(-12) \to (2^2,1^{7}) (-8) 
\to (3,2,1^6) (-6) \oplus (2^5,1) (-5) \\
&\to (4,2,1^5) (-4) \oplus (3,2^3,1^2) (-4) 
\to (3^3,2) (-2) \oplus (4,2^2,1^3) (-3) \\
&\to (4,3^2,1) (-1) \to (4^2,3) \rightarrow I \rightarrow 0
\end{align*}
 
Aside from its compatibility with Ruff's character formula and the special cases handled in the next section, evidence for Conjecture~\ref{conj:BGG conjecture} comes from direct calculation in a number of examples, using the Specht-module-valued version of Jack polynomials defined and studied in~\cite{Gri}. In fact, using the machinery of~\cite{Gri}, one can write down explicit formulas for the maps in our conjectural resolution in many cases; however, as of this writing, we do not have a conjecture for such a formula in general. 

There are a number of other subspace arrangements whose ideals are unitary modules for the rational Cherednik algebra. For instance, see the papers of Feigin~\cite{Fei} and Feigin--Shramov~\cite{FeSh}. One might hope that BGG-style resolutions of these modules also exist, so that one could obtain a  wider class of examples of linear subspace arrangements whose minimal free resolutions are explicitly known. The paper~\cite{Sid} treats some of these arrangements from the point of view of combinatorial commutative algebra.

\section{A proof of the conjecture for 
the $m$-equals set in $\CC^n$ with $2m \ge n+1$}
\label{sec:proof}
  
In this section, we prove Conjecture~\ref{conj:BGG conjecture} under the assumption that $2m \ge n + 1$. Set $A := \Sym(V^*) = \CC[x_1, \dots, x_n]$, $I := I_{1,1,m}$, $X := X_{1,m}$, and $k := n-m+1$. (This conflicts with the notation from the introduction, but should not cause confusion.) Recall that $I$ is the radical ideal of $X$, so $A/I$ has dimension $k$ and degree $\binom{n}{m}$. Also, $A/I$ is Cohen--Macaulay by~\cite[Proposition 3.9]{EGL}. We calculate a resolution of $A/I$ as an $H_c(S_n,\CC^n)$-module by standard modules, which yields a minimal free resolution of $A/I$ considered as an $A$-module. 

A graded free resolution $\bF_\bullet$ (over $A$) is {\it pure} if each $\bF_i$ is generated in a single degree. In this case, the list of these degrees is the {\it degree sequence}. There has been recent interest in modules with pure resolutions in connections with \emph{Boij--S\"oderberg theory}~(\cite{boijsoderberg, es:bs}, see also~\cite{es:survey, floystad} for surveys). We show that $A/I$ has a pure resolution. We will mainly focus on the case $2m > n+1$. The boundary case $2m=n+1$ is easier and is discussed in Remark~\ref{rmk:boundarycase}.

\begin{theorem} 
\label{thm:pureres}
Conjecture~\ref{conj:BGG conjecture} is valid for $\lambda = (m-1,n-m+1)$ and $c = 1/m$ with $m-1 \ge \frac{n}{2}$. In particular, the minimal free resolution of $A/I$ is pure with degree sequence $(0, k, k+1, \dots, n-k, n-k+2, \dots, n-1, n)$. In other words, the resolution of $I$ consists of $2$ linear strands separated by a quadratic jump. 
\end{theorem}

The first construction of pure resolutions was given in~\cite{efw} using the representation theory of the general linear group. The modules $A/I$ could be seen as a special case (we obtain only special kinds of degree sequences) of a conjectural symmetric group analogue of that result. From this perspective, the resolution of $A/I$ is interesting because the sizes of its Betti numbers are as small as possible subject to the given degree sequence.

The minimal resolution we construct in Theorem~\ref{thm:pureres} is also constructed in \cite{Wil}, without using the Cherednik algebra. Unfortunately, it appears that there is a gap in the argument of \cite[Proposition~185]{Wil}. We also point out that the statement of \cite[Proposition~177]{Wil} is not correct in arbitrary characteristic, as shown by the following example. 

\begin{example} 
\label{eg:smallchar}
The ring $A/I$ can fail to be Cohen--Macaulay if we work in positive characteristic. For example, in characteristic $2$ with $m=5$ and $n=7$ (so that our Specht module is $I = S^{(4,3)}$), the graded Betti table of the quotient ring $A/I$ is 
\small 
\begin{quote}
\begin{Verbatim}[samepage=true]
       0  1  2  3 4 5
total: 1 14 21 14 7 1
    0: 1  .  .  . . .
    1: .  .  .  . . .
    2: . 14 21  . . .
    3: .  .  . 14 6 1
    4: .  .  .  . 1 . 
\end{Verbatim}
\end{quote}
\normalsize
using {\tt Macaulay2}~\cite{M2}, and hence the depth of $A/I$ is $2$. In all other characteristics, this depth is $3$. Since the dimension of $A$ is also $3$, in all other cases, $A/I$ is Cohen--Macaulay.
\end{example}

We begin with some preliminary results.
The {\it regularity} of an $A$-module $M$ is 
\[
\reg M := \max_i \{j \mid \Tor^A_i(M,\CC)_{i+j} \ne 0\}.
\]
The module $\Tor^A_i(M,\CC)_{i+j}$ can be identified with the degree $i+j$ generators of the $i$th term in the minimal free resolution of $M$.

\begin{lemma} 
\label{lem:reg}
The regularity of $A/I$ is $k$.
\end{lemma}
\begin{proof}
The Hilbert functions of $A/I$ and $A$ agree in degrees $< k$, so the former is
\[
H_{A/I}(t) = 
\frac{(\sum_{i=0}^{k-1} \binom{n-k + i-1}{i} t^i) + t^kQ(t)}{(1-t)^k}, 
\]
where $Q(t)$ is some polynomial such that $\deg Q + k = \reg A/I$~\cite[Corollary 4.8]{syzygies}. The evaluation of the numerator at $t=1$ is $\deg A/I$, which we know is $\binom{n}{k-1}$. Since the sum on the left of the numerator counts the number of monomials of degree $k-1$ in $n-k+1$ variables, which is $\binom{n-1}{k-1}$, it follows that $Q(1) = \binom{n-1}{k-2}$. From the Hilbert series, the dimension of the space of degree $k$ polynomials in $I$ is $\binom{n-1}{k} - Q(0)$. We also know that this dimension is $\dim S^{(n-k,k)} = \binom{n-1}{k} - \binom{n-1}{k-2}$, and so $Q(0) = Q(1)$. Since $A/I$ is Cohen--Macaulay, $Q(t)$ has nonnegative coefficients, so it follows that $Q(t)$ is a constant polynomial, and hence $\reg A/I = k$.
\end{proof}

Given a partition $\lambda$, let $m_i(\lambda)$ denote the multiplicity of $i$ in $\lambda$. Also, let $\chi^\lambda$ be the character for the Specht module $S^\lambda$. 
Denote by ${\rm ch}$ the Frobenius characteristic map \cite[I.7]{macdonald} that sends $\chi^\lambda$ to the Schur function $s_\lambda$. For notation, we write ${\rm ch}(f) * {\rm ch}(g) = {\rm ch}(f \otimes g)$. Also, we use $s_{\lambda / \mu}$ to denote a skew Schur function \cite[I.5]{macdonald}.

\begin{lemma}
If $Y^i$ is the character of $S_n$ acting on $\Sym^i(\CC^n)$, then 
\begin{align}
\label{eqn:kronecker1} 
{\rm ch}(Y^1) * s_\lambda 
&= s_1 s_{\lambda / 1}, 
\\
\label{eqn:kronecker2} 
{\rm ch}(Y^2) * s_\lambda 
&= s_1s_{\lambda / 1} + s_2 s_{\lambda / 2} + s_{1,1} s_{\lambda / 1,1}, 
\quad \text{and}\\
\label{eqn:kronecker3}
{\rm ch}(Y^3) * s_\lambda 
&= s_3 s_{\lambda / 3} + s_{2,1} s_{\lambda / 2,1} + s_{1,1,1} s_{\lambda / 1,1,1} + (s_2 + s_{1,1})(s_{\lambda/2} + s_{\lambda/1,1}) + s_1 s_{\lambda/1}.
\end{align}
\end{lemma}
\begin{proof}
Equation~\eqref{eqn:kronecker1} is well-known. 
For~\eqref{eqn:kronecker2}, write $p_\lambda$ for the power symmetric functions,  let $n_\lambda$ be the size of the stabilizer in $S_n$ of the conjugacy class of type $\lambda$, and let $1^\mu$ be the character given by $\lambda \mapsto \delta_{\lambda, \mu}$. Since $Y^2(\mu) = \binom{m_1(\mu)}{2} + m_1(\mu) + m_2(\mu)$, 
\begin{align*}
{\rm ch}(Y^2) * p_\lambda 
&= n_\lambda \, {\rm ch}(Y^2) * {\rm ch}(1^\lambda) = n_\lambda \sum_{\mu \vdash n} n_\mu^{-1} Y^2(\mu) 1^\lambda(\mu) p_\mu\\
&= \left(\binom{m_1(\lambda)}{2} + m_1(\lambda) + m_2(\lambda) \right) p_\lambda = \left(\frac{1}{2} p_1^2 \frac{\partial^2}{\partial^2 p_1} + p_1\frac{\partial}{\partial p_1} + p_2 \frac{\partial}{\partial p_2}\right) p_\lambda.
\end{align*}
By~\cite[p.76]{macdonald}, $\frac{\partial}{\partial p_i}$ is the adjoint to multiplication by $p_i / i$ with respect to the Hall inner product \cite[I.4]{macdonald}. Thus since $p_2 = s_2 - s_{1,1}$, 
\[
p_2 \frac{\partial}{\partial p_2} s_\lambda = \frac{1}{2} (s_2 -
s_{1,1})( s_{\lambda / 2} - s_{\lambda / 1,1}).
\]
Also, $p_1 = s_1$ and $s_1^2 = s_2 + s_{1,1}$, so~\eqref{eqn:kronecker2} follows from
\[
p_1^2 \frac{\partial^2}{\partial^2 p_1} s_\lambda = (s_2 +
s_{1,1})(s_{\lambda/2} + s_{\lambda/1,1}).
\]
Using $Y^3(\mu) = m_1(\mu)^2 + m_1(\mu)m_2(\mu) + m_3(\mu)+ \binom{m_1(\mu)}{3}$,~\eqref{eqn:kronecker3} follows similarly. 
\end{proof}

We are now prepared to construct the desired free resolution. 
\begin{proof}[Proof of Theorem~\ref{thm:pureres}]
Let $\bF_0 := A$, and set
\begin{align} 
\label{eqn:firststrand}
\bF_i := S^{(n-k+1-i,k,1^{i-1})} \otimes A(-k+1-i) \quad \text{for \,$i=1,\dots,n-2k+1$.}
\end{align}
The map $\bF_1 \to \bF_0$ is given by the generators of $I$. By~\eqref{eqn:kronecker1}, there is a unique (up to scalar multiple) map $\bF_i \to \bF_{i-1}$ that respects the $S_n$-action for $i=2,\dots,n-2k+1$, and by~\eqref{eqn:kronecker2}, the composition $\bF_{i+2} \to \bF_{i+1} \to \bF_i$ is $0$ for $i = 1, \dots, n-2k-1$. For $i=0$, this composition is $0$ because the minimal degree that $S^{(n-k-1,k,1)}$ appears in $A$ is $k+2$.
Now take 
\begin{align} 
\label{eqn:secondstrand}
\bF_{n-2k+2+i} := S_{(k-1,k-1-i,1^{n-2k+2+i})} \otimes A(-n+k-2-i) \text{ for $i=0,\dots,k-2$}.
\end{align}
By~\eqref{eqn:kronecker2}, there is a unique (up to scalar multiple) map $\bF_{n-2k+2} \to \bF_{n-2k+1}$ that respects the $S_n$-action, and it maps to the kernel of $\bF_{n-2k+1} \to \bF_{n-2k}$ by~\eqref{eqn:kronecker3}.  Again, by~\eqref{eqn:kronecker1} and~\eqref{eqn:kronecker2}, there are unique (up to scalar multiple) maps $\bF_{n-2k+2+i} \to \bF_{n-2k+1+i}$ that respect the $S_n$-action, and these maps together form a complex.

By degree considerations, at each step of the first linear strand given by~\eqref{eqn:firststrand}, the image of the generators of $\bF_i$ in $\bF_{i-1}$ are singular polynomials for $H := H_c(S_n, \CC^n)$. Thus this linear strand is equivariant with respect to $H$ by~\cite[Proposition 3.24]{etingofma}. Furthermore, in the language of~\cite[\S 8.1]{Gri}, the partitions that appear in the first linear strand have diameter $n-k$, so the corresponding standard modules have a unique nonzero proper $H$-submodule. In particular, it must be equal to the image of $\bF_i \to \bF_{i-1}$, so the first linear strand is exact.

We now claim that the images of $\bF_i \to \bF_{i-1}$ for $i=n-2k+2,\dots,n-k$ are the kernels of the previous maps $\bF_{i-1} \to \bF_{i-2}$. We first separately handle the case that $i=n-2k+2$. As mentioned above, the kernel of $\bF_{n-2k+1} \to \bF_{n-2k}$ is the unique nonzero proper $H$-submodule of $\bF_{n-2k+1}$. The map $\bF_{n-2k+2} \to \bF_{n-2k+1}$ is quadratic, so in this case, we must show that there are no linear kernel elements. The only Specht modules that appear in the kernel are obtained by removing some box from $\lambda = (k,k,1^{n-2k+1})$ and adding another box (see~\eqref{eqn:kronecker1}). 
Thus, if there is a linear kernel element coming from, say, the Specht module $S^\mu$, then by~\cite[Corollary 5.2]{Gri}, $\mu \le \lambda$ in dominance order, and $\mu$ and $\lambda$ have the same $m$-core. This can only happen if $\mu = \lambda$, in which case no such copy of $S^\lambda$ generates an $H$-submodule of $\Delta_c(\lambda)$, since $L_c(\lambda)$ appears with multiplicity 1~\cite[Proposition 3.37]{etingofma}.

This same reasoning shows that for $i>n-2k+2$, the image of $\bF_i \to \bF_{i-1}$ contains all linear kernel elements. Since the regularity of $A/I$ is $k$ by Lemma~\ref{lem:reg}, we also conclude that the kernel of $\bF_i \to \bF_{i-1}$ cannot contain any minimal higher degree generators. Thus the second linear strand~\eqref{eqn:secondstrand} is also exact, completing the proof.
\end{proof}

\begin{remark} \label{rmk:boundarycase}
When $n=2m-1$, the quotient by the $m$-equals ideal is still Cohen--Macaulay \cite{EGL} and the ideal is a unitary representation of the Cherednik algebra. This time the partition for the Specht module that generates the ideal is $\lambda = (m-1,m-1,1)$. The minimal free resolution of this ideal is linear, so the verification of Conjecture~\ref{conj:BGG conjecture} can be handled with the methods of this section, but in fact is much simpler in this case.
\end{remark}

The complex constructed in Theorem~\ref{thm:pureres} is self-dual exactly when $k=2$. We conclude with some examples.
The case $n=7$ and $k=2$ yields the following resolution. 

{\scriptsize
\[
\begin{array}{lll} 
0 \to \tiny {\begin{array}{c} \yng(1,1,1,1,1,1,1) \end{array}\!\! (-7)} \searrow
\vspace{-8mm}
  \\ & 
\hspace{-1.5mm}
{\tiny \begin{array}{c} \yng(2,2,1,1,1)\end{array}\!\!(-5) \to
    \begin{array}{c} \yng(3,2,1,1) \end{array}\!\! (-4) \to \begin{array}{c} \yng(4,2,1) \end{array}\!\! (-3) \to
    \begin{array}{c} \yng(5,2) \end{array}\!\! (-2)} \searrow 
\vspace{-5mm}
    \\ & & 
\hspace{-1mm}
    {\tiny \begin{array}{c} \yng(7) \end{array}}
\end{array}
\]
}

\noindent 
Below is the case $n=7$ and $k=3$.
{\scriptsize
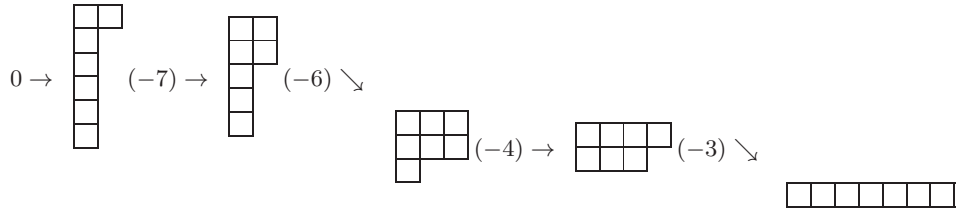
\[
\begin{array}{lll} 0 \to \tiny {\begin{array}{c} \yng(2,1,1,1,1,1) \end{array}\!\! (-7) \to \begin{array}{c} \yng(2,2,1,1,1) \end{array}\!\! (-6) }\searrow
\vspace{-6mm}
    \\
  & 
\hspace{-1mm}
  {\tiny\begin{array}{c} \yng(3,3,1) \end{array}\!\! (-4) \to \begin{array}{c} \yng(4,3) \end{array}\!\! (-3)}  \searrow 
\vspace{-1mm}
  \\ & & 
\hspace{-1.5mm}
  {\tiny \begin{array}{c} \yng(7) \end{array}}
\end{array}
\]
}

\noindent 
Finally, we provide the case $n=9$ and $k=4$.
{\scriptsize
\[
\begin{array}{lll} {\tiny 0 \to \begin{array}{c} \yng(3,1,1,1,1,1,1) \end{array}\!\!\!\! (-9) \to
  \begin{array}{c} \yng(3,2,1,1,1,1) \end{array}\!\!\! (-8) \to \begin{array}{c} \yng(3,3,1,1,1) \end{array}\!\! (-7) } \searrow 
\vspace{-7mm} \\
& 
\hspace{-4mm}
{\tiny \begin{array}{c} \yng(4,4,1) \end{array}\!\! (-5) \to \begin{array}{c} \yng(5,4) \end{array}\!\! (-4)}
\searrow 
\vspace{-1mm}
\\ & &
\hspace{-4mm}
{\tiny \begin{array}{c} \yng(9) \end{array}}
\end{array}
\]
}
\linespread{1}


\end{document}